\documentclass[10pt,journal,compsoc]{IEEEtran}

\usepackage{subfigure}
\usepackage{amsthm}
\usepackage{amsmath,amssymb,amsfonts}
\usepackage{algorithm}
\usepackage{algorithmic}
\usepackage{graphicx}
\usepackage{textcomp}
\usepackage{xcolor}
\usepackage{enumerate}

\newtheorem{definition}{Definition}
\newtheorem{corollary}{Corollary}
\newtheorem{proposition}{Proposition}
\newtheorem{lemma}{Lemma}

\begin{document}
\title{Local Generalization and Bucketization Technique for Personalized Privacy Preservation}

\author{Boyu~Li,
	Kun~He*,~\IEEEmembership{Senior~Member,~IEEE}
	and~Geng~Sun
	\IEEEcompsocitemizethanks{\IEEEcompsocthanksitem B. Li and K. He are with the School of Computer Science and Technology, Huazhong University of Science and Technology, Wuhan, 430074, China.\protect\\
	E-mail: afterslby@163.com, brooklet60@hust.edu.cn
	\IEEEcompsocthanksitem G. Sun is with the College of Computer Science and Technology, Jilin University, Changchun, 130012, China.\protect\\
	E-mail: sungeng207@foxmail.com}
  \thanks{*: Corresponding author}}

\markboth{}
{Shell \MakeLowercase{\textit{et al.}}: Bare Demo of IEEEtran.cls for Computer Society Journals}

\IEEEtitleabstractindextext{%
\begin{abstract}
  Anonymization technique has been extensively studied and widely applied for privacy-preserving data publishing. In most previous approaches, a microdata table consists of three categories of attribute: explicit-identifier, quasi-identifier (QI), and sensitive attribute. Actually, different individuals may have different view on the sensitivity of different attributes. Therefore, there is another type of attribute that contains both QI values and sensitive values, namely, semi-sensitive attribute. Based on such observation, we propose a new anonymization technique, called local generalization and bucketization, to prevent identity disclosure and protect the sensitive values on each semi-sensitive attribute and sensitive attribute. The rationale is to use local generalization and local bucketization to divide the tuples into local equivalence groups and partition the sensitive values into local buckets, respectively. The protections of local generalization and local bucketization are independent, so that they can be implemented by appropriate algorithms without weakening other protection, respectively. Besides, the protection of local bucketization for each semi-sensitive attribute and sensitive attribute is also independent. Consequently, local bucketization can comply with various principles in different attributes according to the actual requirements of anonymization. The conducted extensive experiments illustrate the effectiveness of the proposed approach.
\end{abstract}

\begin{IEEEkeywords}
  Privacy preservation, data publication, local generalization, local bucketization.
\end{IEEEkeywords}}

\maketitle

\IEEEdisplaynontitleabstractindextext

\IEEEpeerreviewmaketitle

\IEEEraisesectionheading{\section{Introduction}\label{sec:introduction}}
\IEEEPARstart{W}{ith} human society enters the age of big data, the variety of individual information, such as income investigation, medical information, and demographic census, has been collected by corporations and governments. These massive personal data are used in data mining and machine learning that contributes to corporations to create business values and governments to develop policies, and also provides people with a more intelligent and convenient way of life. Under the promotion of big data technology, the information barriers between various industries and government departments are gradually broken, and the exchanges and sharings of microdata have become increasingly important activities. However, the published data always contain privacy information. These data may disclose the secrets of individuals if the microdata are published without any disguise.

Many anonymization techniques, such as generalization \cite{Gen} and bucketization \cite{Ana}, are proposed for privacy-preserving data publishing. In these approaches, the attributes in microdata table are classified into three categories: (1) Explicit-Identifier, which can uniquely or mostly identify the record owner and must be removed from the published table; (2) Quasi-Identifier (QI), which can be used to re-identify the record owner when taken together; and (3) Sensitive attribute, which contains the confidential information of individuals.

Generalization transforms the values on QI attributes into general forms, and the tuples whose generalized values are the same constitute an equivalence group. As a result, the records in the same equivalence group are indistinguishable. While bucketization divides the tuples into buckets that breaks the relation between QI attributes and sensitive attributes. Therefore, every record corresponds to the diverse sensitive values within the bucket.

\subsection{Motivation}
\label{sec_moti}
Previous approaches always suppose that an attribute includes only QI values or sensitive values. In fact, different individuals may view different data values as sensitive on the same attribute. Thus, an attribute may contain both QI values and sensitive values, which is considered as semi-sensitive.

For example, a hospital releases some diagnosis records of patients, as shown in Figure \ref{fig_microdata}, to allow researchers to study the characteristics of various diseases. In the microdata table, each attribute except ID has a flag that marks whether a tuple treats her/his value as sensitive (e.g., the tuple with ID 1002 does not care her age value is known by others, but the one with ID 1007 wants to keep secret). In consequence, the attributes of age and zip code are semi-sensitive in the microdata table because they contain both QI values and sensitive values.

\begin{figure*}[!t]
	\centering
	\subfigure[]{\includegraphics[width=3.0in]{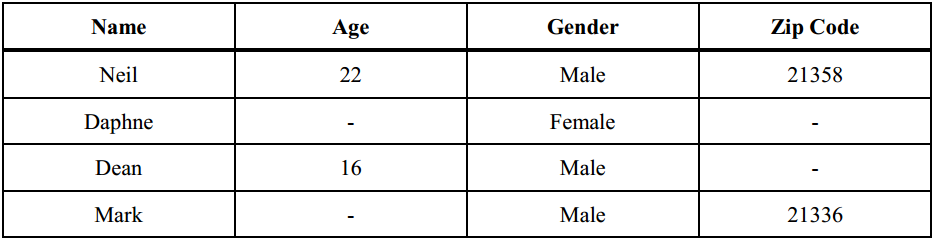}
		\label{fig_background}}
	\hfil
	\subfigure[]{\includegraphics[width=3.0in]{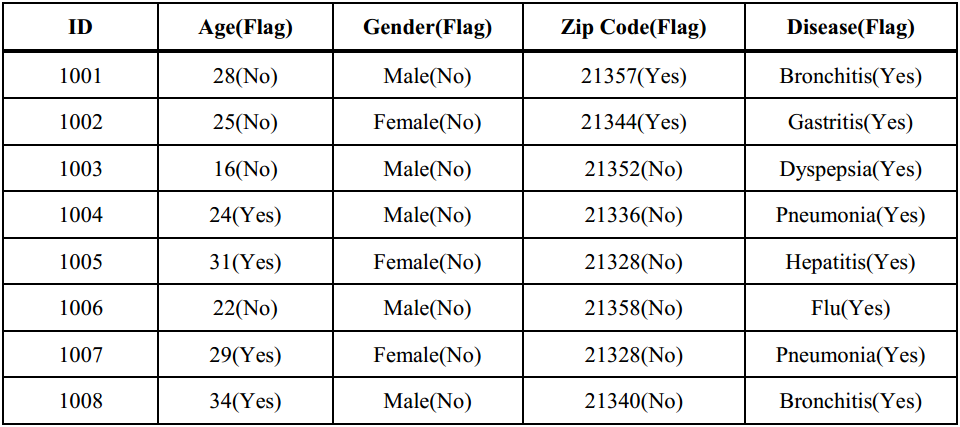}
		\label{fig_microdata}}
	\hfil
	\subfigure[]{\includegraphics[width=3.0in]{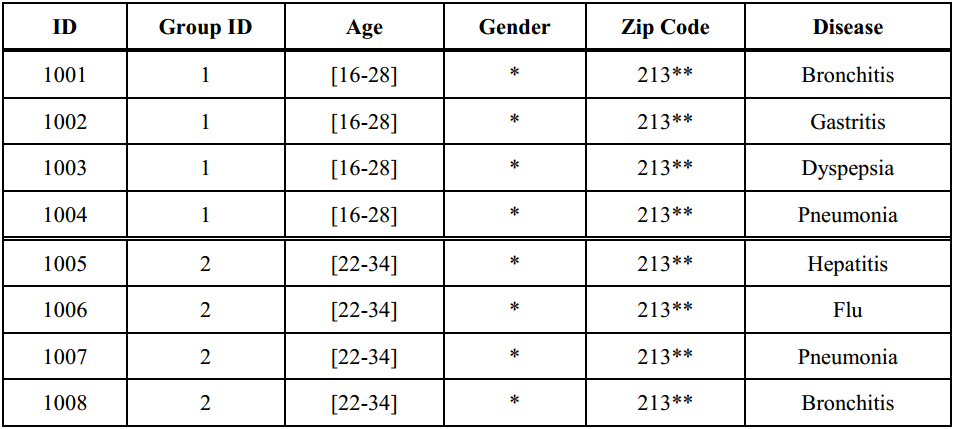}
		\label{fig_gen}}
	\hfil
	\subfigure[]{\includegraphics[width=3.0in]{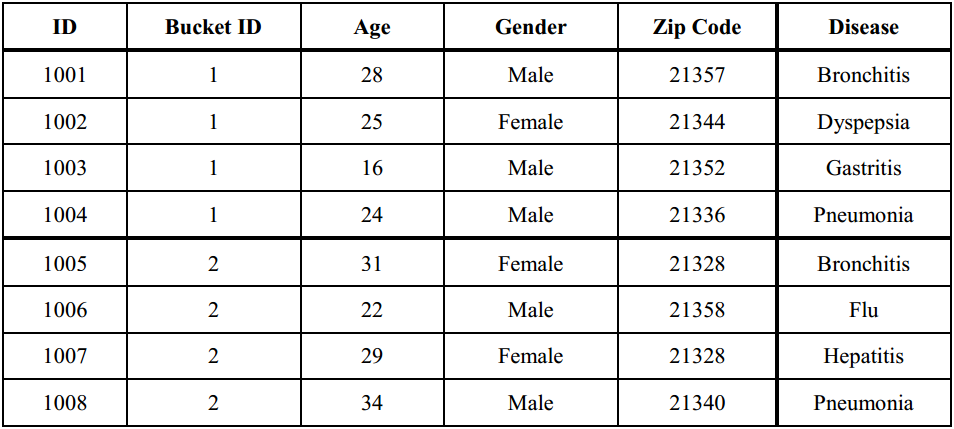}
		\label{fig_ana}}
	\caption{(a) The background knowledge of adversary, (b) the microdata table, (c) the generalized table, (d) the bucketized table}
	\label{fig_1}
\end{figure*}

Suppose that an adversary has the background knowledge as shown in Figure \ref{fig_background} and obtains the microdata table of Figure \ref{fig_microdata}. Knowing that Mark went to the hospital before and matching by his values of gender and zip code, the adversary infers that: (1) his record is with ID 1004 in the microdata table; and (2) his age is 24 and disease is pneumonia. The goal of preventing such privacy disclosures has resulted in the development of many anonymization techniques (see survey \cite{Sur}). Note that, previous generalization and bucketization anonymize whole attributes rather than specific values, then they can only regard semi-sensitive attributes as QI attributes. The generalized and bucketized versions of Figure \ref{fig_microdata} are given in Figure \ref{fig_gen} and \ref{fig_ana}, respectively.

Although generalization effectively prevents identity disclosure, it always suffers from serious information loss as proposed by \cite{Curse, Trade, Injecting}. Almost all the values are irreversibly generalized that hinders recipients from analyzing data information. For example, Figure \ref{fig_gen} stops the adversary from recognizing record owner but in which poor information utility is preserved for recipients.

While the bucketized table preserves excellent information utility, but it only protects the confidential values on the sensitive attribute without caring personalized privacy requirements. The sensitive values on the semi-sensitive attributes tend to be revealed when the adversary has enough background knowledge. For example, the adversary can still acquire the ID and age of Mark by matching QI values in Figure \ref{fig_ana}.

In this paper, we propose local generalization and bucketization (LGB) technique to address the problem of protecting personalized privacy information. LGB prevents both disclosures of identities and sensitive values, and also preserves significant information utility. It uses local generalization and local bucketization to partition the tuples into local equivalence groups in which just specific QI values are generalized and divide the sensitive values into local buckets within each semi-sensitive attribute and sensitive attribute, respectively. The detailed formalization and analysis of LGB are presented in Section \ref{sec_def}. Figure \ref{fig_lgb} shows a possible anonymized result of Figure \ref{fig_microdata} by LGB.

\begin{figure}[!t]
	\centering
	\includegraphics[width=3.0in]{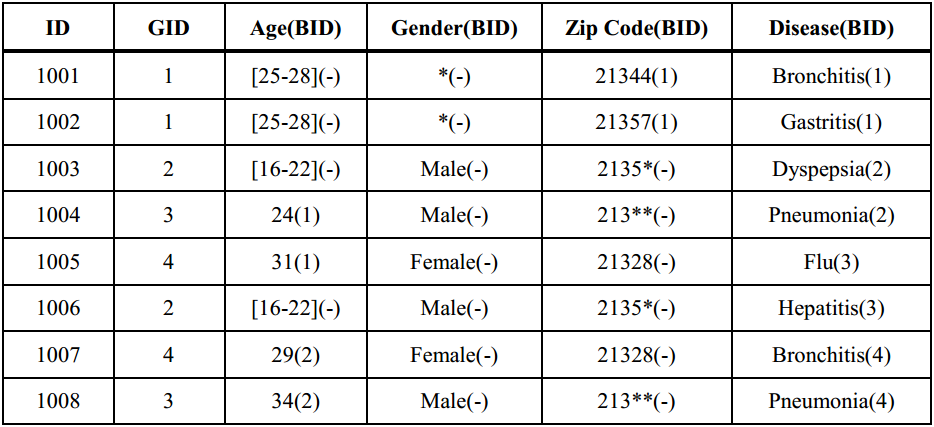}
	\caption{The anonymized table by LGB}
	\label{fig_lgb}
\end{figure}

In Figure \ref{fig_lgb}, the attribute of GID denotes the ID of local equivalence group, and the flag of BID represents the ID of local bucket inside the corresponding attribute. Note that, every local bucket contains only sensitive values, and all the QI values are generalized by local generalization. For example, when the adversary matches the QI values of Mark in Figure \ref{fig_lgb}, he can only infer that Mark's record ID may be 1004 and 1008 which belongs to the local equivalence group of GID 3, and the attribute of age includes the local buckets of BID 1 and 2 in the local equivalence group. Then the adversary concludes that Mark's age may be 24, 31, 29 and 34. For the same reason, the disease value of Mark may be dyspepsia, pneumonia and bronchitis. As a result, the adversary can not determine the exact record ID and sensitive values of target tuple.

\subsection{Contributions}
\label{sec_contri}
This study extends the concept of personalized anonymity \cite{person}. It assumes that individuals can determine their sensitive values at will, an attribute can be QI, semi-sensitive, or sensitive, and a microdata table consists of several QI attributes, semi-sensitive attributes and sensitive attributes. We suppose the background knowledge of adversary is that: (1) the adversary does not acquire any sensitive value on semi-sensitive attribute and sensitive attribute because people must cautiously keep their confidential information from strangers; and (2) in the worst case, the adversary knows the existences and QI values of all the individuals in microdata table. The adversary aims to obtain the ID and sensitive values of target person from the anonymized table. Our contributions are as follows.

First, we propose LGB technique to protect personalized privacy information. LGB combines with local generalization and local bucketization to provide secure protections for identities and sensitive values, and it also reduces information loss as far as possible. The protections of local generalization and local bucketization are independent, so that they can be achieved by appropriate algorithms without weakening the other protection, separately. Additionally, the protection of local bucketization in each semi-sensitive attribute and sensitive attribute is also independent. Therefore, local bucketization can comply with different principles in the different attributes according to the practical demands of anonymization.

Second, we illustrate the effective protections of LGB for identities and sensitive values based on satisfying the principles of $k$-anonymity and $l$-diversity, respectively, i.e., for each tuple, the probabilities of the disclosures of identity and sensitive values are at most 1/$k$ and 1/$l$, respectively. Moreover, since the protections are independent, either degree of protections can be flexibly adjusted according to the actual requirements without reducing the other level of defense, severally.

Third, an efficient algorithm is presented to achieve LGB complying with $k$-anonymity and $l$-diversity. The algorithm contains two main parts, which are local generalization and local bucketization, to partition the tuples into local equivalence groups and divide the sensitive values into local buckets, respectively. We also propose two different algorithms to implement local generalization based on multi-dimensional partition and minimizing normalized certainty penalty (NCP) for different utilization purposes, separately. Furthermore, the range of each local bucket is minimized as far as possible to preserve more information utility.

Last but not least, we conduct a large number of experiments to illustrate the basic property of LGB and the different performances between the two proposed local generalization algorithms through the results of discernibility metric measurement, NCP, and aggregate query answering. The effect of the density of the sensitive values in semi-sensitive attributes is also studied.

The rest of this paper is organized as follows: Section \ref{sec_def} proposes the formalization and analysis of LGB. Section \ref{sec_alg} presents an algorithms to achieve LGB. Section \ref{sec_exp} shows the results and analysis of experiments. Section \ref{sec_rel} describes the related studies. Section \ref{sec_conclu} concludes the paper and proposes the directions for future studies.

\section{LGB}
\label{sec_def}

\begin{figure*}[!t]
	\centering
	\subfigure[]{\includegraphics[width=3.0in]{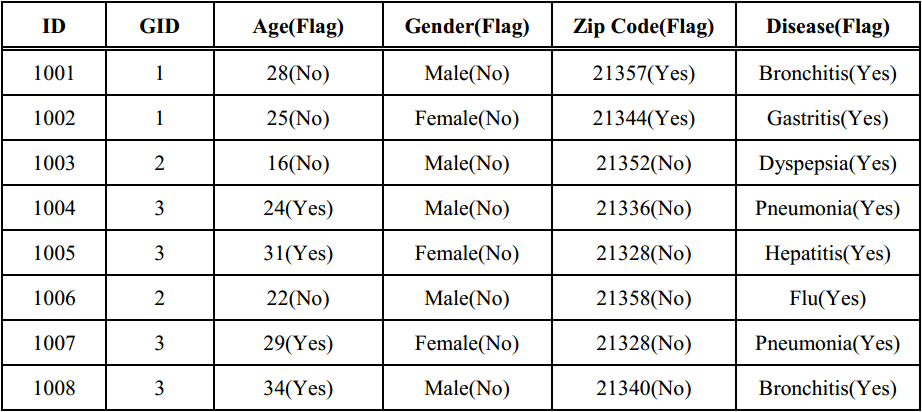}
		\label{fig_qg1}}
	\hfil
	\subfigure[]{\includegraphics[width=3.0in]{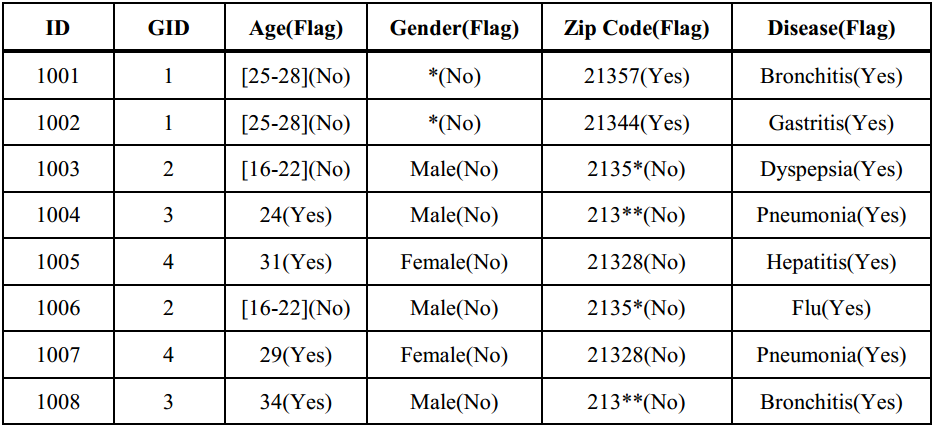}
		\label{fig_qg2}}
	\caption{(a) The first partition, (b) the second partition}
	\label{fig_3}
\end{figure*}

\subsection{Concepts}
\label{sec_concept}
The formalization of LGB technique requires certain prior and novel concepts. We first re-define the categories of attribute based on the property of data value as follows.

\begin{definition}[QI Attribute]
	\label{def_QI}
	An attribute is considered as a QI attribute, denoted as $A^{QI}$, iff the attribute contains only QI values.
\end{definition}

\begin{definition}[Sensitive Attribute]
	\label{def_SA}
	An attribute is considered as a sensitive attribute, denoted as $A^{SA}$, iff the attribute contains only sensitive values.
\end{definition}

\begin{definition}[Semi-Sensitive Attribute]
	\label{def_SS}
	An attribute is considered as a semi-sensitive attribute, denoted as $A^{SS}$, iff the attribute contains both QI values and sensitive values.
\end{definition}

The new definitions of attribute allow individuals to customize their own privacy requirement, and releaser can employ appropriate anonymization approaches to protect people's personalized privacy information and preserve worthy data utilizability according to the characters of different attributes.

\begin{definition}[Partition and QI Group]
	\label{def_par}
	A partition consists of several subsets of $T$, such that each tuple belongs to exactly one subset, and each subset is called a QI group. Specifically, let there be $m$ QI groups $\{G_1, G_2, \cdots , G_m\}$, then $\bigcup^{m}_{i=1}G_i=T$, and for any $1 \leq i_1 \neq i_2 \leq m$, $G_{i_1} \cap G_{i_2}=\emptyset$.
\end{definition}

QI group has different performances when using different anonymization approaches. In generalization, the tuples in the same QI group have the same QI values. While in bucketization, each QI group is divided into two sub-tables, each containing QI values and sensitive values, respectively.

\begin{definition}[Equivalence Group]
	\label{def_equi}
	Given a partition of $T$ with $m$ QI groups, each QI group is called an equivalence group, if for any tuple $t \in T$, a generalized table of $T$ contains the tuple $t$ of the form:
	\begin{center}
		$(G_j[1], G_j[2], \cdots , G_j[d], t[A^{SA}])$,
	\end{center}
	where $G_j(1 \leq j \leq m)$ is the unique QI group including $t$, $G_j[i](1 \leq i \leq d)$ is the generalized value on $A^{QI}_i$ for all the tuples in $G_j$, and $t[A^{SA}]$ represents $t$'s value on $A^{SA}$.
\end{definition}

\begin{definition}[Bucket]
	\label{def_buc}
	Given a partition of $T$ with $m$ QI groups, each QI group is called a bucket, if each QI group is represented as the form:
	\begin{center}
		$QIT(QI, BID)$ and $SAT(SA, BID)$,
	\end{center}
	where $QI$ and $SA$ are the QI values and sensitive values of the tuples in the QI group, respectively, and $BID$ denotes the ID of bucket.
\end{definition}

In previous equivalence groups, all the whole attributes including QI values are generalized to the same form that always causes overprotection. We propose local generalization technique based on the new definitions of attribute to partition the tuples into local equivalence groups through generalizing just specific QI values.

\begin{definition}[QI Partition]
	\label{def_QIp}
	For any tuple $t \in T$, $QI[t]$ represents the set of the attributes containing $t$'s QI values, such that
	\begin{center}
		$QI[t]=\{A | t[A]\ is\ a\ QI\ value\}$.
	\end{center}
	A QI partition of $T$ divides the table into disjoint subsets $\{T_{1}, T_{2}, \cdots, T_{m}\}$, such that for any $1 \leq i \neq j \leq m$, $T_{i} \cap T_{j}=\emptyset$, $\bigcup^{m}_{i=1}T_i=T$, and for any $t_{i_{1}}, t_{i_{2}} \in T_{i}$, $QI[t_{i_{1}}] = QI[t_{i_{2}}]$.
\end{definition}

\begin{definition}[Local Equivalence Group]
	\label{def_le}
	Given a microdata table $T$ and a QI partition of $T$ with $m$ subsets, each subset is called a local equivalence group, if the QI values are generalized to the same form in the corresponding attribute, such that for any tuple $t \in T$, a locally generalized table of $T$ contains the tuple $t$ of the form:
	\begin{center}
		$(LEG_j[A^{QI}_1], \cdots , LEG_j[A^{QI}_p], t[A^{SA}_1], \cdots , t[A^{SA}_q])$,
	\end{center}
	where $LEG_j(1 \leq j \leq m)$ is the unique local equivalence group including $t$, $A^{QI}_{i_{1}}(1 \leq i_{1} \leq p)$ and $A^{SA}_{i_{2}}(1 \leq i_{2} \leq q)$ denote the attributes containing QI value and sensitive value of $t$ in $LEG_j$, respectively, $LEG_j[A^{QI}]$ is the generalized value on $A^{QI}$ for all the tuples in $LEG_j$, and $t[A^{SA}]$ represents $t$'s value on attribute $A^{SA}$.
\end{definition}

Local generalization contains two partition steps. First, it divides the tuples into subsets by QI partition in which all the records carry the QI values on the same attributes. For example, in Figure \ref{fig_qg1}, the tuples with ID 1001 and 1002 are in the same subset because both of them just carry QI values on the attributes of age and gender. Next, local generalization partitions the tuples into local equivalence groups within each subset, and generalizes their QI values. For example, the tuples in the subset of GID 3 in Figure \ref{fig_qg1} are divided into the local equivalence groups of GID 3 and 4 in Figure \ref{fig_qg2}, and every group in Figure \ref{fig_qg2} is a local equivalence group.

Likewise, previous bucketization protects for the whole sensitives attribute rather than specific sensitive values. We present local bucketization technique to partition the sensitive values into local buckets in the corresponding attribute.

\begin{definition}[Local Bucket]
	\label{def_lb}
	For any semi-sensitive attribute or sensitive attribute in $T$, the sensitive values are partitioned into local buckets, and each local bucket has the form:
	\begin{center}
		$IDT(ID, BID)$ and $SAT(SA, BID)$,
	\end{center}
	where $ID$ and $SA$ represent the IDs and sensitive values of the tuples in the local bucket, respectively, and $BID$ denotes the ID of local bucket within the attribute.
\end{definition}

\begin{figure}[!t]
	\centering
	\includegraphics[width=3.0in]{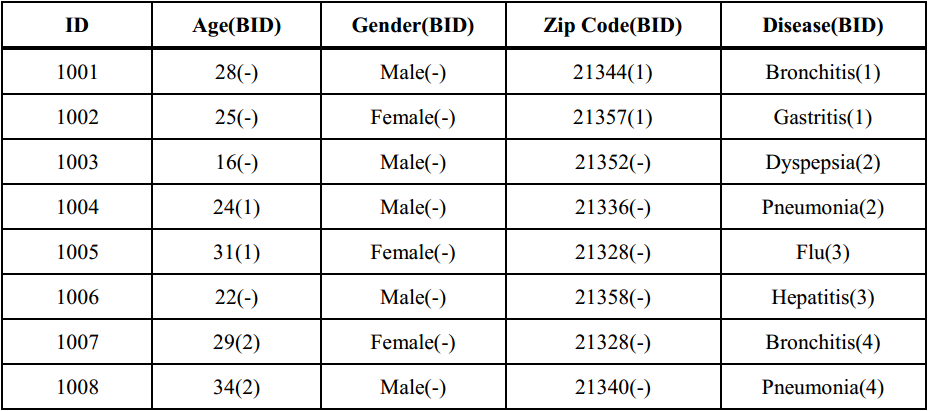}
	\caption{The locally bucketized table}
	\label{fig_lb}
\end{figure}

For example, in Figure \ref{fig_lb}, the tuples with ID 1004 and 1005 are in the same local bucket of BID 1 within the attribute of age, but they are in the different local buckets of BID 2 and 3 within the attribute of disease. Therefore, the local buckets within the different attributes are independent.

Note that, the previous equivalence groups and buckets can be treated as the special cases of local equivalence groups and local buckets when the microdata table does not contain any semi-sensitive attribute, respectively. Based on Definition \ref{def_le} and \ref{def_lb}, we define LGB technique as follows.

\begin{definition}[Local Generalization and Bucketization]
	\label{def_lgb}
	Given a microdata table $T$, a local generalization and bucketization of $T$ is given by the partitions of local generalization and local bucketization, and every tuple and sensitive value belongs to exactly one local equivalence group and local bucket, respectively.
\end{definition}

\subsection{Protection Analysis}
\label{sec_pa}
In this section, we analyze the protections of local generalization and local bucketization against the disclosures of identities and sensitive values in detail. Without loss of generality, we illustrate how local generalization and local bucketization comply with $k$-anonymity \cite{Kanony} and $l$-diversity \cite{Ldiver}, respectively. Then, we prove the locally generalized and bucketized table can also satisfy $k$-anonymity and $l$-diversity by meeting the corresponding conditions. We first consider the protection against identity disclosure in the locally generalized table and have the following lemma and corollary.

\begin{lemma}
	\label{lem_lg}
	Given a locally generalized table, for any tuple $t \in T$, the probability of identity disclosure is at most $1/|LEG(t)|$, where $LEG(t)$ is the local equivalence group including $t$.
\end{lemma}

\begin{proof}
	According to Definition \ref{def_le}, the tuples in the same local equivalence group have the same attributes containing QI values, and all the QI values are generalized to the same form. Consequently, the adversary must obtain at least $|LEG(t)|$ possible tuples by matching QI values, then the probability of identity disclosure is at most $1/|LEG(t)|$.
\end{proof}

\begin{corollary}
	\label{cor_lg}
	A locally generalized table complies with $k$-anonymity, if every local equivalence group contains at least $k$ tuples.
\end{corollary}

\begin{proof}
	Given a locally generalized table in which every local equivalence group includes at least $k$ tuples, for any tuple $t \in T$, we have
	\begin{center}
		$|LEG(t)| \ge k$,
	\end{center}
	where $|LEG(t)|$ is the size of the local equivalence group including $t$. Then
	\begin{center}
		$\frac{1}{|LEG(t)|} \le \frac{1}{k}$.
	\end{center}
	According to Lemma \ref{lem_lg}, the probability of identity disclosure for any tuple is at most $1/k$. Therefore, the locally generalized table complies with $k$-anonymity.
\end{proof}

Next, we discuss the protection for sensitive values in the locally bucketized table. Suppose an adversary knows $t$'s existence and QI values, then attempts to infer $t$'s sensitive value $s$ from the locally bucketized table. The adversary needs to find $t$'s possible records in the locally bucketized table by matching the QI values.

\begin{definition}[Matching Tuple]
	\label{def_mt}
	Given a locally bucketized table $T^{buc}$, and for any tuple $t \in T$, a tuple $mt \in T^{buc}$ is a matching tuple of $t$ if each QI value of $t$ matches that of $mt$.
\end{definition}

\begin{definition}[Matching Bucket]
	\label{def_mb}
	Given a locally bucketized table $T^{buc}$, and for any tuple $t \in T$, a local bucket $mb$ within an attribute is a matching bucket of $t$ if there is at least a matching tuple of $t$ in $mb$.
\end{definition}

For example, when the adversary matches Mark in the locally bucketized table of Figure \ref{fig_lb}, he can infer that Mark's matching tuple is with ID 1004, and the matching buckets inside the attributes of age and disease are with ID 1 and 2, respectively.

We denote $p(t,s)$ as the probability that the sensitive value $s$ of $t$ is exposed, and let $p(t,b)$ represent the probability that $t$ is in the bucket $b$. Then we have the following lemma and corollary.

\begin{lemma}
	\label{lem_lb}
	Given a locally bucketized table, for any tuple $t \in T$, the probability that any sensitive value $s$ of $t$ is exposed is as follows:
	\begin{center}
		$p(t,s) \leq \sum_{mb}p(t,mb)\frac{|mb(s')|}{|mb|}$,
	\end{center}
	where $|mb(s')|$ is the number of the most occurrence sensitive value $s'$ in the matching bucket $mb$, and $|mb|$ is the size of $mb$.
\end{lemma}

\begin{proof}
	To acquire the sensitive value $s$, the adversary has to calculate the probabilities that $t$ exists in each local bucket and $t$ carries the sensitive value $s$ within each local bucket. Then, the adversary has:
	\begin{center}
		$p(t,s)=\sum_{B}p(t,b)p(s|t,b)$,
	\end{center}
	where $p(s|t,b)$ denotes the probability that $t$ carries the sensitive value $s$ given that $t$ is in the local bucket $b$. The adversary eliminates the local bucket that does not contain any matching tuple of $t$, expressed as follows:
	\begin{center}
		$p(t,b)=0$, if $\nexists mt \in b$.
	\end{center}
	According to Definition \ref{def_mb}, we have:
	\begin{center}
		$p(t,s)=\sum_{mb}p(t,mb)p(s|t,mb)$,
	\end{center}
	The most occurrence sensitive value $s'$ in $mb$ is expressed as:
	\begin{center}
		$|mb(s)| \leq |mb(s')|$.
	\end{center}
	Thus:
	\begin{center}
		$p(s|t,mb)=\frac{|mb(s)|}{|mb|} \leq \frac{|mb(s')|}{|mb|}$,
	\end{center}
	then:
	\begin{center}
		$p(t,s)=\sum_{mb}p(t,mb)p(s|t,mb) \leq \sum_{mb}p(t,mb)\frac{|mb(s')|}{|mb|}$.
	\end{center}
\end{proof}

\begin{corollary}
	\label{cor_lb}
	A locally bucketized table complies with $l$-diversity principle, if every local bucket satisfies the conditions: (1) each sensitive value appears at most once in the local bucket; and (2) the size of each local bucket is at least $l$.
\end{corollary}

\begin{proof}
	According to Lemma \ref{lem_lb}, for any tuple $t \in T$, we have
	\begin{center}
		$p(t,s) \leq \sum_{mb}p(t,mb)\frac{|mb(s')|}{|mb|}$.
	\end{center}
	We confine that each sensitive value appears at most once inside the local bucket, such that for any $s \in T$,
	\begin{center}
		$|mb(s)| \leq 1$.
	\end{center}
	And for any local bucket $b$, we have:
	\begin{center}
		$|b| \geq l$.
	\end{center}
	Then:
	\begin{center}
		$p(t,s) \leq \sum_{mb}p(t,mb)\frac{|mb(s')|}{|mb|} \leq \frac{1}{l}\sum_{mb}p(t,mb)=\frac{1}{l}$.
	\end{center}
	In consequence, the locally bucketized table complies with $l$-diversity by meeting the conditions.
\end{proof}

Finally, we prove that a locally generalized and bucketized table complies with $k$-anonymity and $l$-diversity through satisfying the conditions in Corollary \ref{cor_lg} and \ref{cor_lb}.

\begin{corollary}
	\label{cor_lgb}
	A locally generalized and bucketized table complies with $k$-anonymity and $l$-diversity by meeting the conditions as follows: (1) each local equivalence group contains at least $k$ tuples; (2) each sensitive value appears at most once inside each local bucket; and (3) the size of each local bucket is at least $l$.
\end{corollary}

\begin{proof}
	According to Definition \ref{def_lgb}, the protections of local generalization and local bucketization are independent, and the conditions in Corollary \ref{cor_lg} and \ref{cor_lb} are non-overlapping. As a result, as long as the locally generalized and bucketized table satisfies the corresponding conditions, it complies with $k$-anonymity and $l$-diversity.
\end{proof}

Generally, local generalization enhances the protection of local bucketization because local generalization transforms QI values into the same form that increases the number of the matching tuples of target tuple. For example, the tuple with ID 1006 is in the local bucket of BID 3 within the attribute of disease in Figure \ref{fig_lb}, while he is in that of BID 2 and 3 in Figure \ref{fig_lgb}, because his QI values are generalized to the same form in the local equivalence group of GID 2 that increases the number of his matching tuples. Then the probability that his disease value is disclosed is decreased to $1/4$.

\section{Algorithms}
\label{sec_alg}
This section presents an algorithm to achieve LGB complying with $k$-anonymity and $l$-diversity. In addition, two algorithms are proposed to implement local generalization for different utilization purposes. The main procedure of LGB is given in Algorithm \ref{alg_lgb}.

\begin{algorithm}[!h]
	\caption{LGB($T$, $k$, $l$)}
	\label{alg_lgb}
	\begin{algorithmic}[1]
		\STATE $Attri_{sen}=\{the\ attributes\ including\ sensitive\ values\}$
		\STATE $T_{anony}=T$
		\FOR {\bf{each} $attr \in Attri_{sen}$}
		\STATE $ValuePair_{sen}=\{(id,s)|s \in attr\ and\ s\ is\ sensitive\}$
		\STATE $local\_bucketization(T_{anony},ValuePair_{sen},l)$
		\ENDFOR
		\STATE $local\_generalization(T_{anony},k)$
		\STATE \bf{return} \rm $T_{anony}$
	\end{algorithmic}
\end{algorithm}

The data structure $Attri_{sen}$ (line 1) stores the attributes including sensitive values in $T$, i.e., the set of semi-sensitive attributes and sensitive attributes. The variable $T_{anony}$ (line 2) denotes the anonymized result, and it is initialized as $T$. In each iteration (lines 3 to 6), the algorithm picks an attribute from $Attri_{sen}$ and chooses the tuples with sensitive values (line 4). Then the algorithm divides the tuples into local buckets based on the value of $l$ (line 5). After the loop, the function $local\_generalization(T_{anony},k)$ divides $T_{anony}$ into local equivalence groups according to the value of $k$ (line 7). Finally, the algorithm returns $T_{anony}$ as the anonymized result of $T$ (line 8). Note that, in Algorithm \ref{alg_lgb}, the QI attributes are not contained in $Attri_{sen}$, and $ValuePair_{sen}$ does not include any tuple with QI value in $attr$. Therefore, none of local buckets contains any QI value.

The procedure comprises two main parts: local bucketization (line 5) and local generalization (line 7). We elaborate each part in the rest of this section.

\subsection{Local Bucketization}
\label{sec_lb}
This section presents an efficient algorithm to implement the function $local\_bucketization(T_{anony},ValuePair_{sen},l)$ in Algorithm \ref{alg_lgb}. To preserve more information utility, the algorithm partitions the tuples in $ValuePair_{sen}$ into local buckets and minimizes the range of the sensitive values in each local bucket as far as possible. The detailed procedure is shown in Algorithm \ref{alg_lb}.

\begin{algorithm}[!h]
	\caption{local\_bucketization($T$, $ValuePair$, $l$)}
	\label{alg_lb}
	\begin{algorithmic}[1]
		\STATE $value\_number=\{(value, number)|count \ in \ ValuePair\}$
		\STATE $median=calculate\_median(value\_number)$
		\STATE $VP_{small}=\{(id, s)|(id, s) \in ValuePair \ and \ s \le median\}$
		\STATE $VP_{big}=\{(id, s)|(id, s) \in ValuePair \ and \ s > median\}$
		\IF {$check\_condition(VP_{small},l)$ $\bf{and}\rm$ $check\_condition(VP_{big},l)$}
		\STATE $local\_bucketization(T,VP_{small},l)$
		\STATE $local\_bucketization(T,VP_{big},l)$
		\ELSE
		\STATE $divide\_buckets(T,ValuePair,l)$
		\ENDIF
	\end{algorithmic}
\end{algorithm}

The algorithm recursively divides $ValuePair$ into two smaller sets, and their ranges of sensitive values are non-overlapping. The data structure $value\_number$ stores every sensitive value and the number of occurrence counted in $ValuePair$ (line 1). The algorithm calculates the median value of the sensitive values based on their numbers in $value\_number$ (line 2), and divides $ValuePair$ into two smaller sets (lines 3 and 4). The function $check\_condition(ValuePair,l)$ checks whether $ValuePair$ can be divided into local buckets complying with $l$-diversity (line 5), such that the product of the number of the most occurrence sensitive value and the value of $l$ is not more than the size of $ValuePair$. If both $VP_{small}$ and $VP_{big}$ meet the condition, the algorithm makes recursive calls of $local\_bucketization(T,VP_{small},l)$ and $local\_bucketization(T,VP_{big},l)$ (lines 6 and 7). Otherwise, the algorithm divides the tuples in $ValuePair$ into local buckets (line 9). The function $divide\_buckets(T,ValuePair,l)$ is implemented by the assignment algorithm of $m$-Invariance \cite{Minvar} to satisfy the conditions in Corollary \ref{cor_lb}.

\begin{proposition}
	$T_{anony}$ complies with $l$-diversity after local bucketization phase.
\end{proposition}

\begin{proof}
	The function $divide\_buckets(T,ValuePair,l)$ is implemented by the assignment algorithm of $m$-Invariance, where the parameters $l$ and $m$ are mathematically equivalent. In each recursion of Algorithm \ref{alg_lb}, both sets of $VP_{small}$ and $VP_{big}$ are checked to satisfy $l$-eligible condition \cite{Minvar}. According to $m$-Invariance, the assignment algorithm divides the tuples into $m$-unique buckets. Then the size of each generated local bucket is at least $l$, and every sensitive value appears at most once in each local bucket. As a result, all the local buckets within the corresponding attribute meet the conditions in Corollary \ref{cor_lb} based on the value of $l$. After the loop (lines 3 to 6) in Algorithm \ref{alg_lgb}, the tuples with sensitive values inside each semi-sensitive attribute and sensitive attribute are partitioned into $l$-unique local buckets, so that $T_{anony}$ satisfies $l$-diversity principle.
\end{proof}

\subsection{Local Generalization}
\label{sec_lg}
In this section, we propose two algorithms to implement local generalization based on multi-dimensional partition and minimizing NCP, separately. The experiments we conducted in Section \ref{sec_exp} will show their different performances. We elaborate each algorithm as follows.

\subsubsection{Based on Multi-Dimensional Partition}
\label{sec_mul}
Previous multi-dimensional partition \cite{Mon} is an effective and popular approach to divide the tuples into equivalence groups. However, it does not apply to the publishing scenario of personalized privacy requirement. We combine multi-dimensional partition and QI partition to divide the tuples into local equivalence groups according to their specific QI values.

\begin{definition}[Multi-Dimensional QI Partition]
	\label{def_mdp}
	Given a microdata table $T$ with $d$ attributes and a QI partition with $m$ subsets $\{T_{1}, T_{2}, \cdots, T_{m}\}$, a multi-dimensional QI partition divides the tuples into non-overlapping multi-dimensional regions within each subset that covers $D[A^{i}_{1}] \times D[A^{i}_{2}] \times \cdots \times D[A^{i}_{d}] (1 \leq i \leq m)$, where $D[A^{i}_{j}] (1 \leq j \leq d)$ denotes the domain of attribute $A_{j}$ inside subset $T_{i}$, and for any tuple $t \in T$, $(t[A_{1}], t[A_{2}], \cdots, t[A_{d}])$ is mapped in the unique region.
\end{definition}

To reduce information loss as far as possible, the size of each local equivalence group should be minimized to only satisfy $k$-anonymity, so that each region is divided into smaller ones until at least one of their sizes is less than $k$. Algorithm \ref{alg_lgm} describes the local generalization based on multi-dimensional partition in detail.

\begin{algorithm}
	\caption{local\_generalization$(T,k)$}
	\label{alg_lgm}
	\begin{algorithmic}[1]
		\STATE $T_{anony}=\emptyset$
		\STATE $T_{subsets}=\{the\ subsets\ of\ T\ divided\ by\ QI\ partition\}$
		\FOR {\bf{each} $T_{set} \in T_{subsets}$}
		\STATE $partition\_set=\{T_{set}\}$
		\WHILE {$partition\_set \neq \emptyset$}
		\STATE $oper\_set=pick\_set(partition\_set)$
		\STATE $partition\_set=partition\_set-oper\_set$
		\STATE $QI\_set = QI[oper\_set]$
		\STATE $partition\_flag = false$
		\WHILE {$QI\_set \neq \emptyset$}
		\STATE $attri=choose\_dimension(oper\_set,QI\_set)$
		\STATE $split\_value=cal\_median(oper\_set,attri)$
		\STATE $S_{l}=\{t|t \in oper\_set\ and\ t[attri] \leq split\_value\}$
		\STATE $S_{r}=\{t|t \in oper\_set\ and\ t[attri] > split\_value\}$
		\IF {$|S_{l}| \geq k$ $\bf{and}\rm$ $|S_{r}| \geq k$}
		\STATE $partition\_set=partition\_set+S_{l}$
		\STATE $partition\_set=partition\_set+S_{r}$
		\STATE $partition\_flag=true$
		\STATE $\bf{break}$
		\ELSE
		\STATE $QI\_set=QI\_set-attri$
		\ENDIF
		\ENDWHILE
		\IF {$partition\_flag\ is\ false$}
		\STATE $gen\_set=generalize(oper\_set)$
		\STATE $T_{anony}=T_{anony}+gen\_set$
		\ENDIF
		\ENDWHILE
		\ENDFOR
		\STATE \bf{return} $T_{anony}$
	\end{algorithmic}
\end{algorithm}

The data structures $T_{anony}$ and $T_{subsets}$ store the anonymized result and the subsets of $T$ divided by QI partition, respectively (lines 1 and 2). In each iteration (lines 3 to 29), the algorithm picks and divides a subset $T_{set}$ into local equivalence groups. The data structure $partition\_set$ contains the sets of tuples which has not been generalized, and it includes $T_{set}$ in the beginning (line 4). As long as $partition\_set$ is not empty (line 5), the algorithm picks and eliminates a set from $partition\_set$ (lines 6 and 7). The data structure $QI\_set$ denotes the set of the attributes including QI values in $oper\_set$ (line 8), and $partition\_flag$ represents whether $oper\_set$ can be divided (line 9). While $QI\_set$ is not empty (line 10), the algorithm chooses an attribute $attri$ and calculates the median value (lines 11 and 12), then divides $oper\_set$ into two smaller sets (lines 13 and 14). If the sizes of $S_{l}$ and $S_{r}$ are both bigger than or equal to $k$ (line 15), the algorithm adds $S_{l}$ and $S_{r}$ into $partition\_set$ (lines 16 and 17) and sets $partition\_flag$ as $true$ (line 18), then breaks the while loop (line 19). Otherwise, $attri$ is eliminated from $QI\_set$ (line 21). If $partition\_flag$ is false after the while loop, not a single attribute can be used to divide $oper\_set$ into the smaller sets that complies with $k$-anonymity (line 24). The algorithm generalizes the QI values in $oper\_set$ and adds the generalized set $gen\_set$ into $T_{anony}$ (lines 25 and 26). Finally, the algorithm returns $T_{anony}$ as the generalized result (line 30).

Local generalization based on multi-dimensional partition is well suited as a common approach because it evenly divides the tuples into local equivalence groups that reduces much information loss. In practice, some microdata tables are published for particular purposes, and the information utility should be evaluated by given metric. Next, we present another algorithm to achieve local generalization which preserves information utility evaluated by specific metric as far as possible.

\subsubsection{Based on Minimizing NCP}
\label{sec_ncp}
In this section, we propose a utility-based algorithm to implement local generalization through using NCP \cite{Ncp} as information metric. For any tuple $t \in T$, $t$'s value $v$ is generalized to $v^{*}$ on the categorical attribute $A_{cat}$, then 
\begin{center}
	$NCP_{A_{cat}}(t)=\frac{size(v^{*})}{|A_{cat}|}$,
\end{center}
where $size(v^{*})$ is the number of the leaf nodes that are the descendants of $v^{*}$ in the hierarchy tree of $A_{cat}$, and $|A_{cat}|$ denotes the number of distinct values on attribute $A_{cat}$. While $t$'s value $v$ is generalized to $[v^{*}_{lower}, v^{*}_{upper}]$ on the numeric attribute $A_{num}$, then 
\begin{center}
	$NCP_{A_{num}}(t)=\frac{v^{*}_{upper}-v^{*}_{lower}}{range(A_{num})}$,
\end{center}
where $v^{*}_{lower}$ and $v^{*}_{upper}$ are the lower bound and upper bound of generalized range, respectively, and $range(A_{num})$ is the range of all the values on attribute $A_{num}$. The information loss of whole generalized table is represented as
\begin{center}
	$NCP(T)=\sum_{t \in T}\sum_{i=1}^{d}NCP_{A_{i}}(t)$.
\end{center}

The local generalization complying with $k$-anonymity based on minimizing NCP is described as Algorithm \ref{alg_lgn}. The data structures $T_{anony}$ and $T_{subsets}$ store the anonymized result and the subsets of $T$ divided by QI partition, respectively (lines 1 and 2). In each iteration (lines 3 to 24), the algorithm picks and divides a subset $T_{set}$ into local equivalence groups (line 3). The data structure $QI\_set$ denotes the set of the attributes including QI values in $T_{set}$ (line 4), and $partition\_set$ represents the sets of tuples which has not been generalized (line 5). While $partition\_set$ is not empty (line 6), the algorithm picks and eliminates a set $oper\_set$ from $partition\_set$ (lines 7 and 8). If the size of $oper\_set$ is smaller than $2k$ (line 9), the tuples in $oper\_set$ are generalized and added into $T_{anony}$ (lines 10 and 11). Otherwise, the function $find\_seeds(oper\_set, QI\_set)$ returns two seed records that maximizes the value of NCP based on $QI\_set$ (line 13), and the algorithm divides $oper\_set$ into two smaller sets according to the seed tuples (line 14). The functions $find\_seeds(oper\_set, QI\_set)$ and $divide\_table(oper\_set, sd_{1}, sd_{2}, QI\_set)$ can be implemented by the top-down algorithm in \cite{Ncp}. Next, if the sizes of $T_{1}$ and $T_{2}$ are both more than or equal to $k$, the algorithm adds them into $partition\_set$ (lines 15 to 17), or else, $oper\_set$ is generalized and added into $T_{anony}$ (lines 18 to 20). Finally, the algorithm returns $T_{anony}$ as the generalized result (line 25).

\begin{algorithm}
	\caption{local\_generalization$(T,k)$}
	\label{alg_lgn}
	\begin{algorithmic}[1]
		\STATE $T_{anony}=\emptyset$
		\STATE $T_{subsets}=\{the\ subsets\ of\ T\ divided\ by\ QI\ partition\}$
		\FOR {\bf{each} $T_{set} \in T_{subsets}$}
		\STATE $QI\_set=QI[T_{set}]$
		\STATE $partition\_set=\{T_{set}\}$
		\WHILE {$partition\_set \neq \emptyset$}
		\STATE $oper\_set=pick\_set(partition\_set)$
		\STATE $partition\_set=partition\_set-oper\_set$
		\IF {$|oper\_set|<2k$}
		\STATE $gen\_set=generalize(oper\_set)$
		\STATE $T_{anony}=T_{anony}+gen\_set$
		\ELSE
		\STATE $sd_{1}, sd_{2}=find\_seeds(oper\_set, QI\_set)$
		\STATE $T_{1}, T_{2}=divide\_table(oper\_set, sd_{1}, sd_{2}, QI\_set)$
		\IF {$|T_{1}| \geq k$ $\bf{and}\rm$ $|T_{2}| \geq k$}
		\STATE $partition\_set=partition\_set+T_{1}$
		\STATE $partition\_set=partition\_set+T_{2}$
		\ELSE
		\STATE $gen\_set=generalize(oper\_set)$
		\STATE $T_{anony}=T_{anony}+gen\_set$
		\ENDIF
		\ENDIF
		\ENDWHILE
		\ENDFOR
		\STATE \bf{return} $T_{anony}$
	\end{algorithmic}
\end{algorithm}

\begin{proposition}
	$T_{anony}$ complies with $k$-anonymity after local generalization phase.
\end{proposition}

\begin{proof}
	In both Algorithm \ref{alg_lgm} and \ref{alg_lgn}, the size of subset must be more than or equal to $k$ prior to putting it into $partition\_set$, and only $oper\_set$ selected from $partition\_set$ is generalized and added into $T_{anony}$. Therefore, the sizes of all the local equivalence groups are at least $k$ that satisfies the condition in Corollary \ref{cor_lg}. Then $T_{anony}$ satisfies $k$-anonymity principle.
\end{proof}

\section{Experiments}
\label{sec_exp}
This section evaluates the efficiency of LGB technique. We use the real US Census data \cite{data}, eliminate the tuples with missing values, and randomly select 31,055 tuples with nine attributes. The QI attributes are relationship, marital status, race, education and hours per week, the semi-sensitive attributes are sex, age and occupation, and the sensitive attribute is salary. Table \ref{tab_attri} describes the attributes in detail.

\begin{table}[!t]
	\caption{Description of the attributes}
	\label{tab_attri}
	\centering
	\begin{tabular}{|c|c|c|c|c|}
		\hline
		& \bf{Attribute} & \bf{Value Type} & \bf{Sensitivity Type} & \bf{Size}\\
		\hline
		1 & Sex & Categorical & Semi-sensitive & 2\\
		\hline
		2 & Age & Continuous & Semi-sensitive & 73\\
		\hline
		3 & Relationship & Categorical & QI & 13\\
		\hline
		4 & Marital status & Categorical & QI & 6\\
		\hline
		5 & Race & Categorical & QI & 9\\
		\hline
		6 & Education & Categorical & QI & 11\\
		\hline
		7 & Hours per week & Continuous & QI & 93\\
		\hline
		8 & Occupation & Categorical & Semi-sensitive & 257\\
		\hline
		9 & Salary & Continuous & Sensitive & 719\\
		\hline
	\end{tabular}
\end{table}

In Section \ref{sec_alg}, we propose an algorithms to implement LGB that satisfies $k$-anonymity to prevent identity disclosure and $l$-diversity to protect sensitive values. However, the attribute of sex contains only two distinct sensitive values, namely, male and female, so that complying with $l$-diversity principle can not provide effective protection. We severally employ $t$-closeness \cite{tc} to confine the attribute of sex and $l$-diversity to protect the rest semi-sensitive attributes and sensitive attribute. In addition, the algorithms based on multi-dimensional partition and minimizing NCP are denoted as LGB\_MDP and LGB\_NCP, respectively. The experiments will show: (1) the basic property of LGB technique and the different performances between LGB\_MDP and LGB\_NCP; and (2) the effect of the density of sensitive values in the semi-sensitive attributes.

\subsection{Information Metrics}
\label{sec_exp_utility}
In this section, the experiments employ two information metrics, namely, discernibility metric measurement \cite{penal} and NCP, to check the information utility. The parameters $k$ is assigned to 5, 8 and 10, and $l$ is assigned to 5, 8, 10, 12, 15, 18 and 20, respectively, and $t$ is fixed at 0.2. Besides, the percentage of the sensitive values in each semi-sensitive attribute is set to about 20\%.

The first experiment uses discernability metric measurement, denoted as $C_{DM}$, which is given by the equation:
\begin{center}
	$C_{DM}=\sum_{LEG}|LEG|^2$,
\end{center}
where $LEG$ is the local equivalence group, and $|LEG|$ denotes the size of $LEG$. A smaller value of $C_{DM}$ means less generalization and perturbation in the anonymization process. Figure \ref{fig_pen_mon_2} and \ref{fig_pen_td_2} present the results of the anonymized table by LGB\_MDP and LGB\_NCP, respectively.

\begin{figure*}[!t]
	\centering
	\subfigure[]{\includegraphics[width=3.0in]{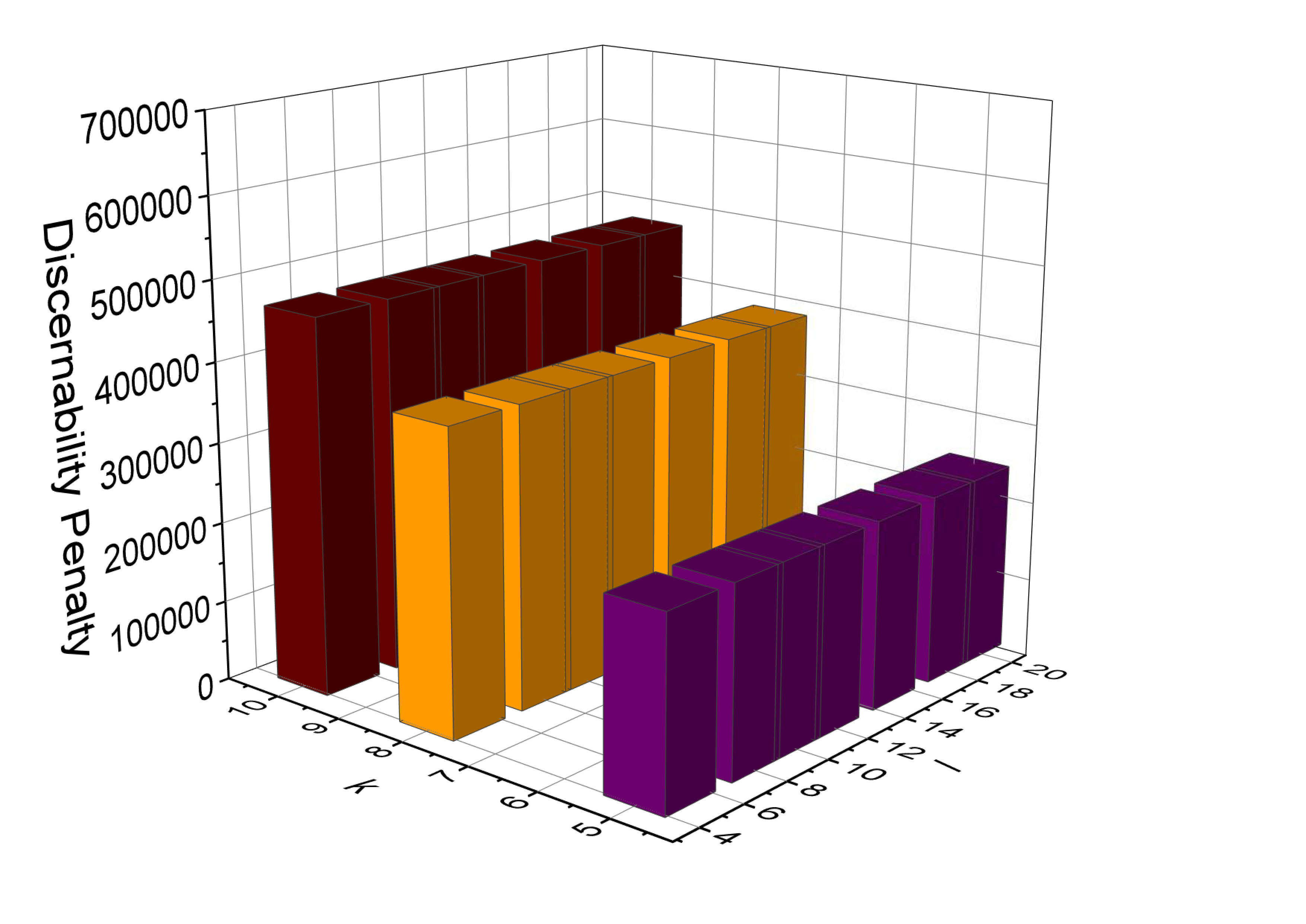}
		\label{fig_pen_mon_2}}
	\hfil
	\subfigure[]{\includegraphics[width=3.0in]{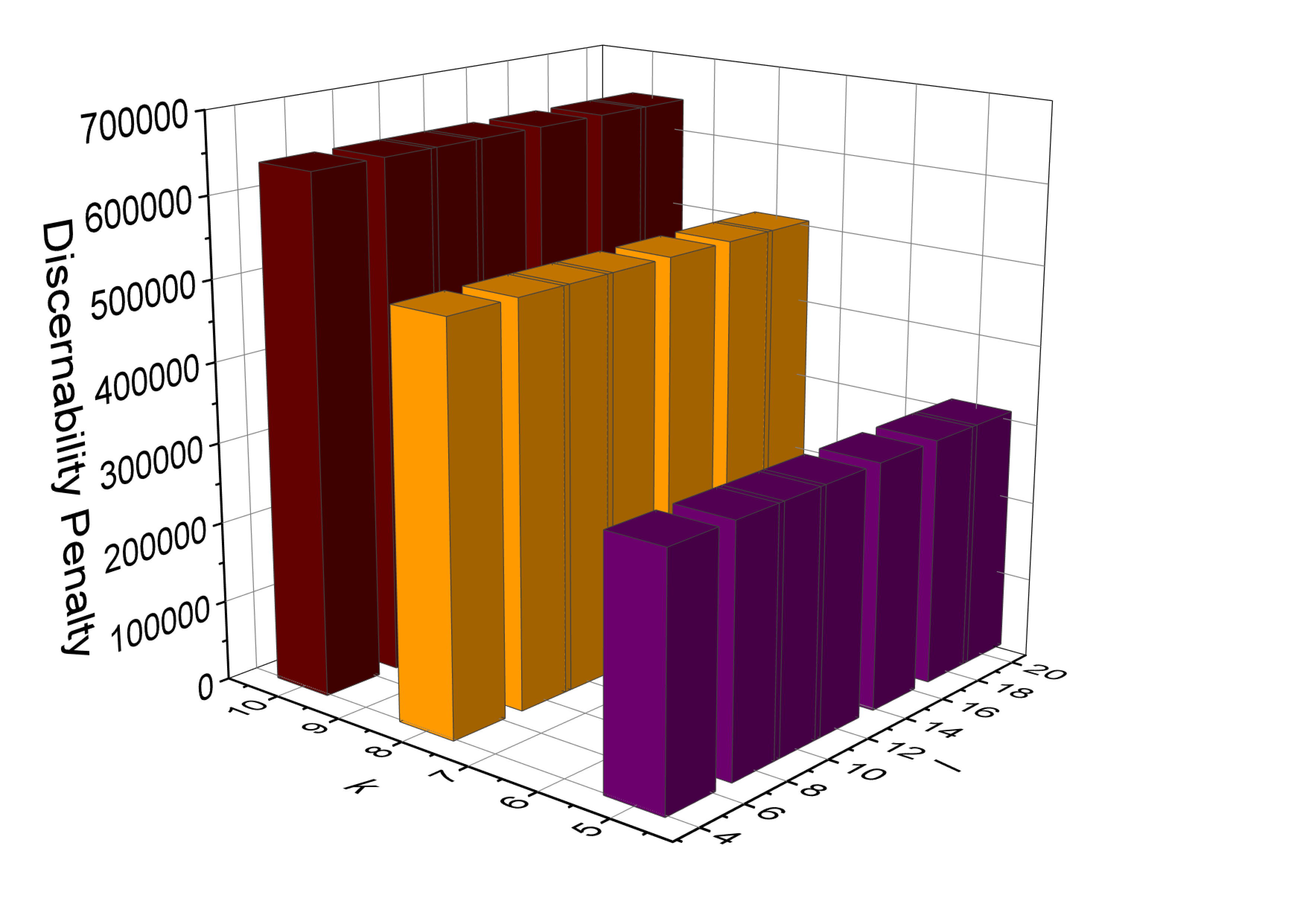}
		\label{fig_pen_td_2}}
	\caption{Discernibility metric results}
	\label{fig_pen}
\end{figure*}

It is obvious that the discernability penalty of LGB\_MDP is much lower than that of LGB\_NCP, so that LGB\_MDP divides the tuples into local equivalence groups more evenly than LGB\_NCP. Besides, the sizes of local equivalence groups are hardly affected by $l$ because the values of discernability penalty are the same when the value of $k$ is fixed. Therefore, the results indicate that the protection of local generalization is barely affected by local bucketization.

Next, the experiment uses NCP as information metric to compare the information quality between LGB\_MDP and LGB\_NCP. Figure \ref{fig_NCP_mon_2} and \ref{fig_NCP_td_2} present the results of LGB\_MDP and LGB\_NCP, respectively. It is shown that the values of LGB\_NCP are clearly lower than that of LGB\_MDP by about 14\% to 17\% although LGB\_MDP divides the tuples into local equivalence groups more evenly than LGB\_NCP. Thus, the utility-based heuristic of LGB\_NCP plays a role in reducing NCP. It proves that LGB is a flexible framework which can be implemented by appropriate algorithms to meet different actual demands.

\begin{figure*}[!t]
	\centering
	\subfigure[]{\includegraphics[width=3.0in]{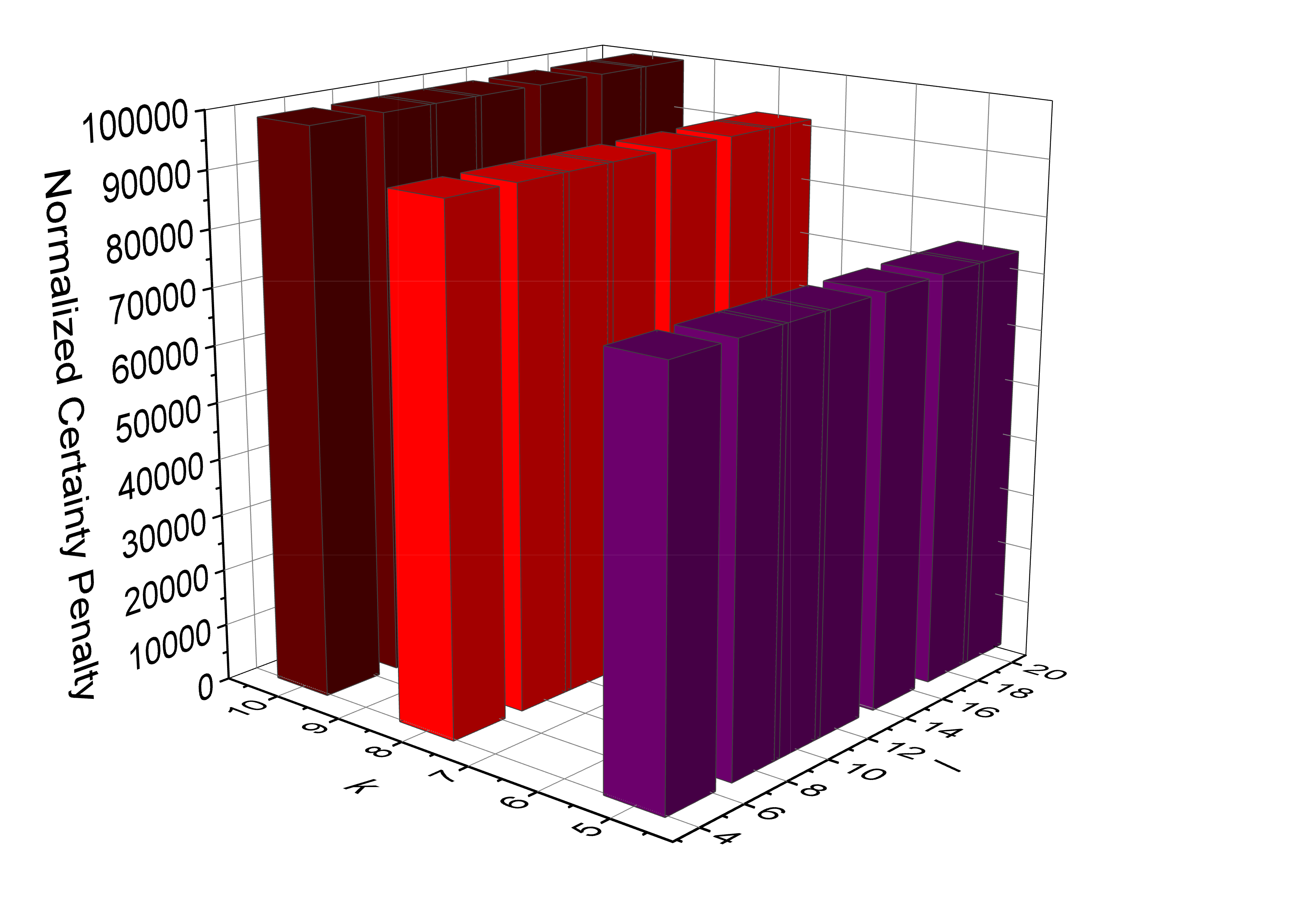}
		\label{fig_NCP_mon_2}}
	\hfil
	\subfigure[]{\includegraphics[width=3.0in]{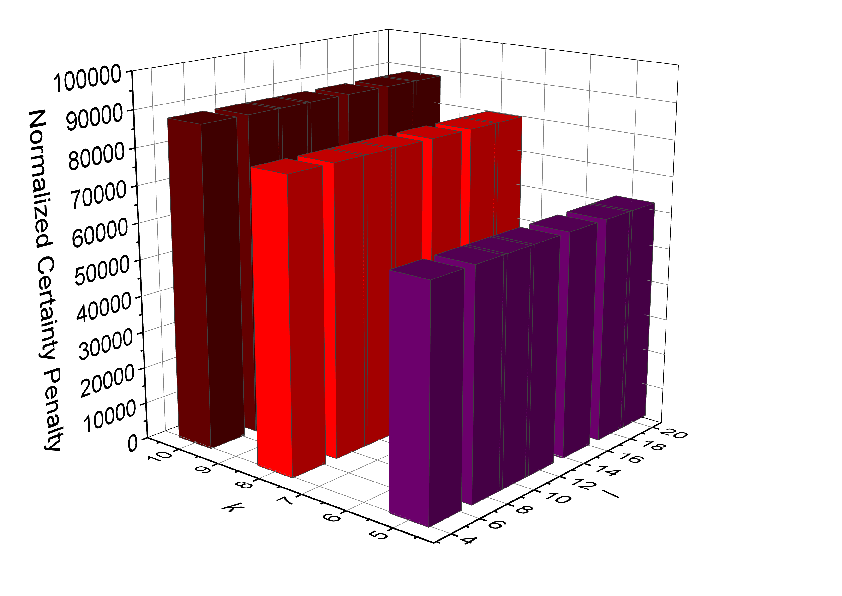}
		\label{fig_NCP_td_2}}
	\caption{NCP results}
	\label{fig_ncp}
\end{figure*}

\subsection{Query Answering}
\label{sec_exp_query}
In this experiment, we use the approach of aggregate query answering \cite{per} to check information utility. We randomly generate 1,000 queries and calculate the average relative error for each anonymized table. The sequence of query is expressed as the form:
\begin{flushleft}
	SELECT SUM(salary) FROM \bf{Microdata}\rm\\
	WHERE \bf{$pred$($A_1$)} AND \bf{$pred$($A_2$)} AND \bf{$pred$($A_3$)} AND \bf{$pred$($A_4$)}.
\end{flushleft}
Specifically, the query condition contains four random QI attributes or semi-sensitive attributes, and the sum of salary is the result for comparison. For categorical attributes, the predicate $pred(A)$ has the following form:
\begin{center}
	$(A=v_1\ or\ A=v_2\ or\ \cdots\ or\ A=v_m)$,
\end{center}
where $v_i(1 \leq i \leq m)$ is a random value from $D[A]$. While for numerical attributes, the predicate $pred(A)$ has the following form:
\begin{center}
	$(A>v)\ or\ (A<v)\ or\ (A=v)$\\ $or\ (A \geq v)\ or\ (A \leq v)\ or\ (A \neq v)$,
\end{center}
where $v$ is a random value from $D[A]$. According to \cite{per}, the relative error rate, denoted as $R_{error}$, is given by the equation:
\begin{center}
	$R_{error}=(Sum_{upper}-Sum_{lower})/Sum_{act}$,
\end{center}
where $Sum_{upper}$ and $Sum_{lower}$ are the upper bound and lower bound of the sum of salary, respectively, and $Sum_{act}$ is the actual value.

\begin{figure*}[!t]
	\centering
	\subfigure[]{\includegraphics[width=3.0in]{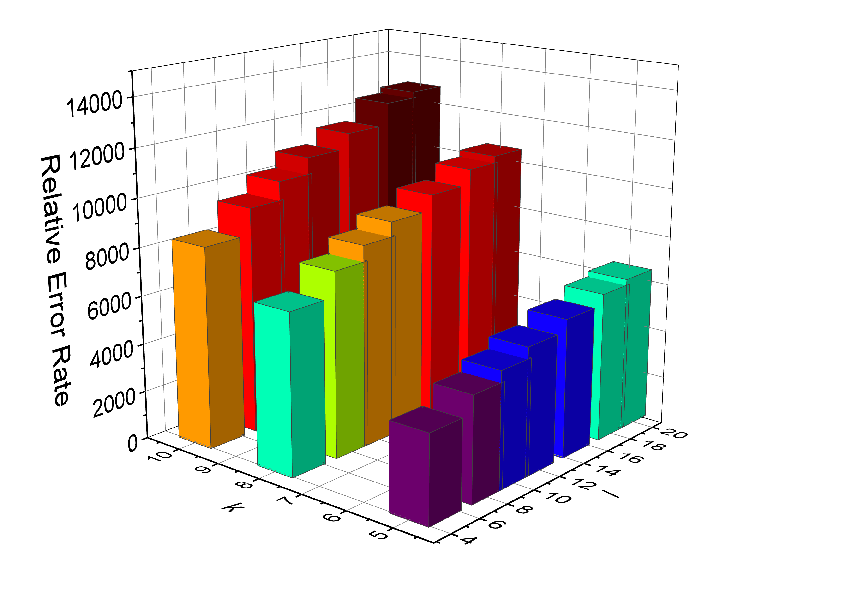}
		\label{fig_query_mon_2}}
	\hfil
	\subfigure[]{\includegraphics[width=3.0in]{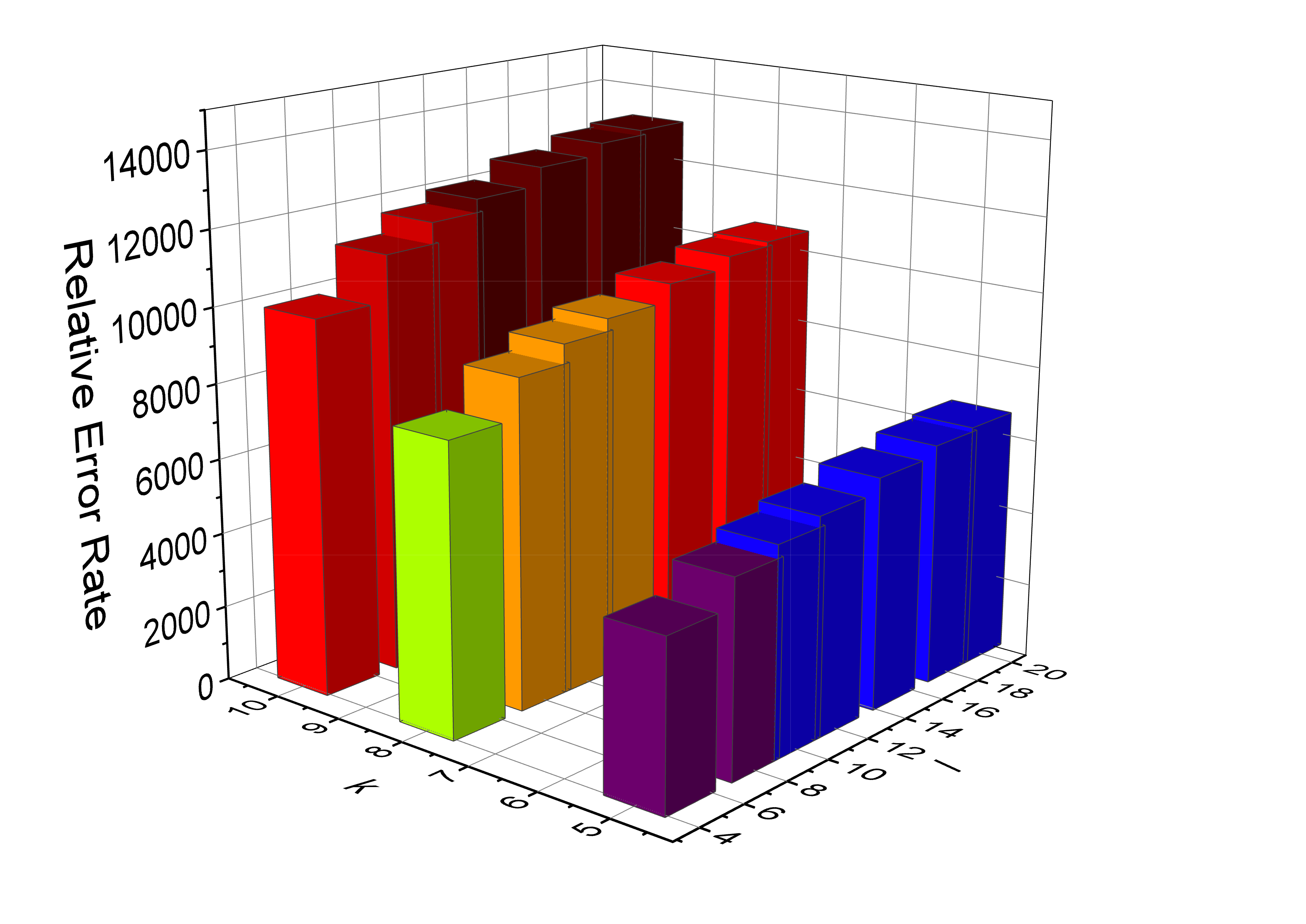}
		\label{fig_query_td_2}}
	\caption{Query answering results}
	\label{fig_query}
\end{figure*}

\begin{figure*}[!t]
	\centering
	\subfigure[]{\includegraphics[width=3.0in]{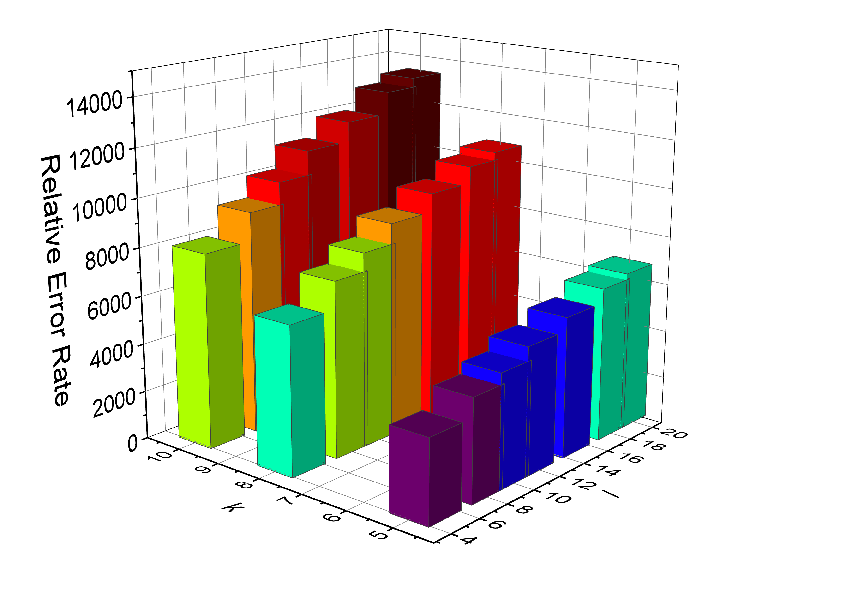}
		\label{fig_query_mon_1}}
	\hfil
	\subfigure[]{\includegraphics[width=3.0in]{Query_mon_2}
		\label{fig_query_mon_}}
	\hfil
	\subfigure[]{\includegraphics[width=3.0in]{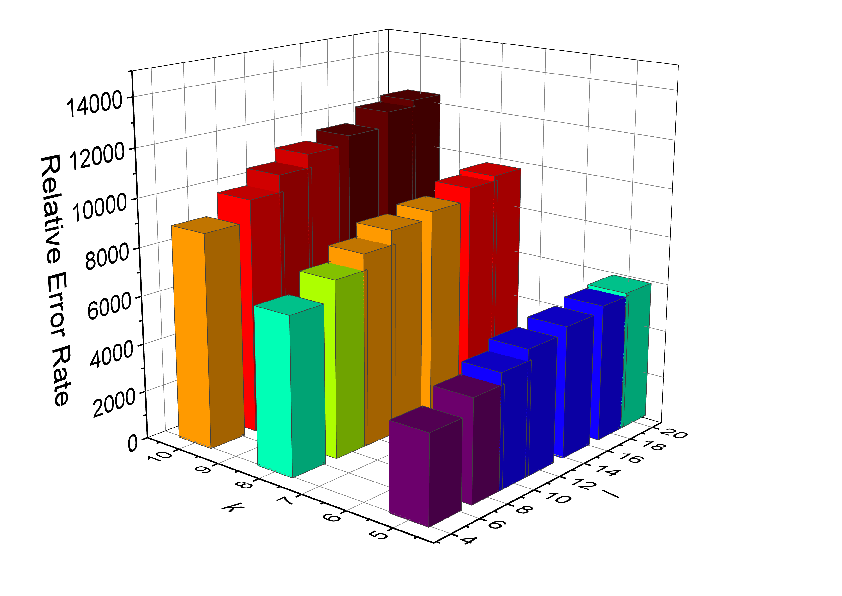}
		\label{fig_query_mon_3}}
	\hfil
	\subfigure[]{\includegraphics[width=3.0in]{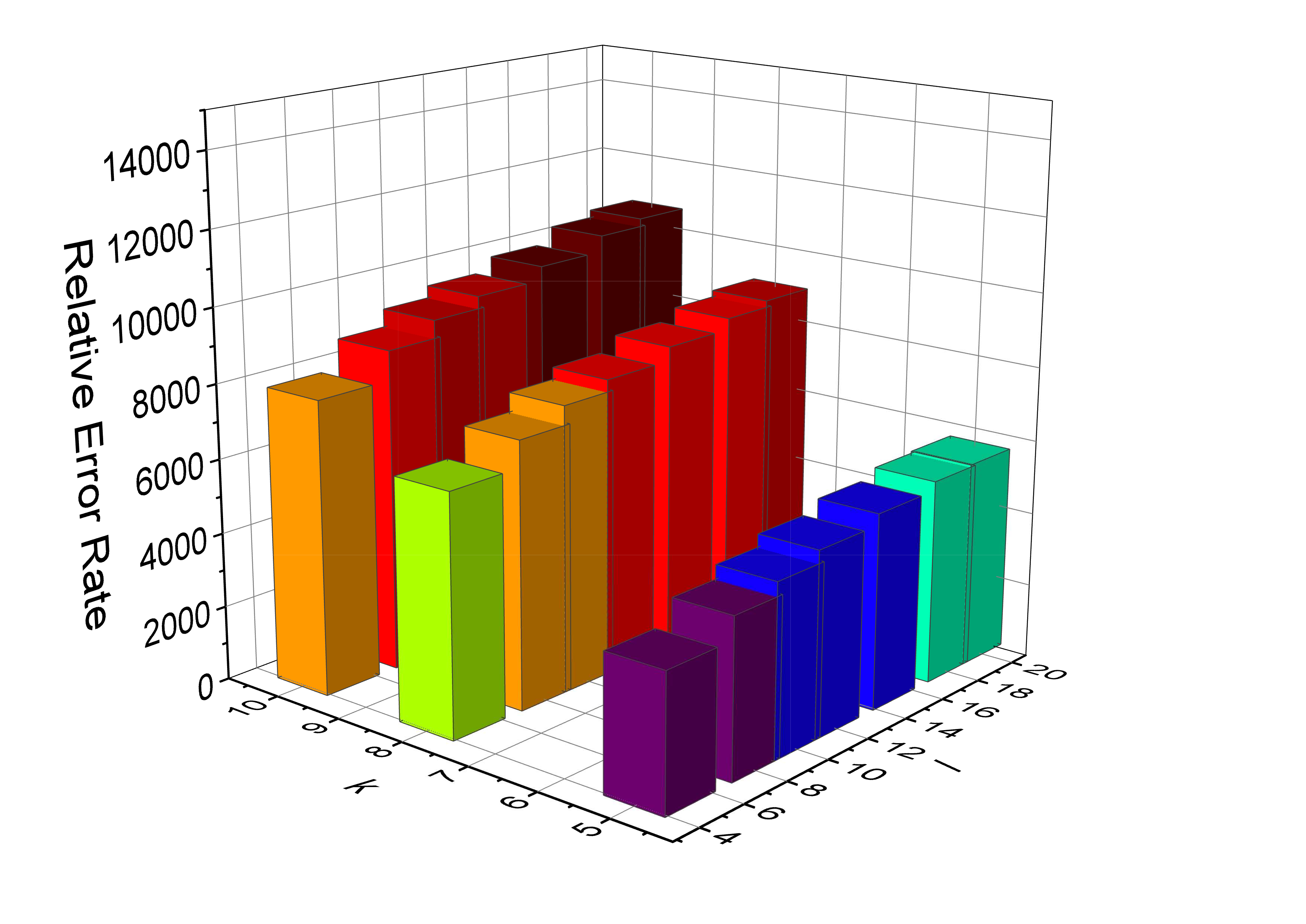}
		\label{fig_query_mon_4}}
	\caption{Different density of sensitive values}
	\label{fig_effect}
\end{figure*}

According to \cite{per}, $Sum_{upper}$ and $Sum_{lower}$ are calculated through ''numbers of hits'' in the group. However, each tuple is contained in the local buckets of different semi-sensitive attributes, we need to count the upper bound and lower bound of the possible numbers of the tuples who satisfies the query conditions in each local bucket of salary. Let $Pro_{t}(A)$ denote the probability that $t$ matches the condition on attribute $A$, and $Pro(t)$ represents the probability that $t$ satisfies the query conditions, then we have
\begin{center}
	$Pro(t)=\prod_{A} Pro_{t}(A)$.
\end{center}
For each local bucket within salary, we calculate the sum of probabilities, then round down and round up the sum as the lower bound and upper bound of the number of tuples, respectively. Next, we count the lower bound and upper bound of the sum of salary in the local bucket through help table, denoted as $Sum_{lower}(b)$ and $Sum_{upper}(b)$, respectively. The lower bound and upper bound of the sum of salary in the whole table are expressed as
\begin{center}
	$Sum_{lower}=\sum_{B} Sum_{lower}(b)$,
\end{center}
and
\begin{center}
	$Sum_{upper}=\sum_{B} Sum_{upper}(b)$.
\end{center}

Figure \ref{fig_query_mon_2} and \ref{fig_query_td_2} show the results of aggregate query answering of LGB\_MDP and LGB\_NCP, respectively. The relative error rate increases with the growth of the value of $k$ or $l$, and it is more affected by $k$ than that of $l$. Consequently, local generalization has more influence on information utility than local bucketization. Note that, LGB\_MDP has lower relative error rates than LGB\_NCP when $l$ is small, because LGB\_MDP divides the tuples into local equivalence groups more evenly than LGB\_NCP that narrows the ranges of generalized values. But with the growth of $l$, the relative error rates of LGB\_MDP is gradually close to that of LGB\_NCP, because the ranges of local buckets are large enough to cover the query conditions, then local bucketization has more and more important influence on information utility. In conclusion, LGB\_MDP is better than LGB\_NCP if the value of $l$ is small, otherwise, LGB\_MDP and LGB\_NCP perform very equal.

\subsection{Effect of Density}
\label{sec_exp_density}
In this experiment, we check the effect of the density of the sensitive values in semi-sensitive attributes on LGB\_MDP. The percentages of sensitive values are set to about 10\%, 20\%, 30\% and 40\% in each semi-sensitive attribute, respectively, and the rest configurations are the same as the previous experiments. The results are evaluated through aggregate query answering, and Figure \ref{fig_query_mon_1}, \ref{fig_query_mon_}, \ref{fig_query_mon_3}, and \ref{fig_query_mon_4} present the results of the percentages of 10\%, 20\%, 30\% and 40\%, respectively.

With the density of sensitive values increases, the relative error rate is declined. The most obvious reduction is in Figure \ref{fig_query_mon_4}, where the percentage of sensitive values is about 40\%, and the values of $k$ and $l$ are 10 and 20, respectively. It is mainly because more values are divided into local buckets rather than that of equivalence groups, and the values are indistinguishable in equivalence groups but are specific in local buckets, so that less QI values narrow the ranges of generalized values, and more sensitive values do not loss the accuracies in local buckets and even compose more local buckets in which the ranges are further decreased.

\section{Related Work}
\label{sec_rel}
Privacy-preserving data publishing determines an optimal trade-off between privacy protection and information preservation. Above all, the most important factor in applying appropriate anonymization technique is assuming the publishing scenario. Most research focus on anonymizing a single and static release so far. However, there are more complicated data publishing environments in actual applications, such as multiple release publishing \cite{seq}, continuous data publishing \cite{contin1, contin2}, and collaborative data publishing \cite{colla}. Moreover, personalized privacy preserving \cite{person} is also an important publishing scenario because publishers always ignore the concrete needs of individuals which may causes serious privacy disclosures.

Secondly, an anonymity principle should be chosen or proposed to provide secure protections against privacy disclosures according to the assumed background knowledge and purpose of adversary. Generally, the privacy threats can be divided into four categories, which are membership disclosure, identity disclosure, attribute disclosure, and probabilistic disclosure. Many useful anonymity principles can be used to prevent these privacy attacks. $k$-Anonymity \cite{Kanony} is one of the most powerful and widespread principles for protecting sensitive information, especially for defending identity attack. $l$-Diversity \cite{Ldiver} and $t$-closeness \cite{tc} are popular and effective principles for preventing attribute disclosure. $\delta$-Presence \cite{presen} and $(d, \gamma)$-privacy \cite{probabi} hinder membership disclosure and probabilistic disclosure, respectively. Note that, according to Dwork \cite{diff}, the absolute privacy protection is impossible. Therefore, a robust anonymity framework is supposed to prevent most leakages of confidential information through combining various principles. Additionally, differential privacy \cite{diff_sur}, which is an effective model for protecting statistic data, does not need to assume the background knowledge of adversary. It prevents different disclosures of queries through adding noise based on different mechanisms \cite{mech1, mech2, mech3, mech4, mech5}.

Finally, an explicit algorithm should be presented that complies with anonymity principle and minimizes information loss as far as possible. For example, slicing \cite{slic} complying with $l$-diversity principle protects sensitive attributes, prevents attribute disclose and membership disclosure to some extent, and preserves better data utility than generalization. Microaggregation \cite{microag} is a flexible technique that satisfies $k$-anonymous $t$-closeness through merging clusters, and it can also adjust the priority of $k$-anonymity and $t$-closeness according to actual demands, such that microaggregation can choose to preserve more information utility by $k$-anonymity-first or provide more information security by $t$-closeness-first. Cross-bucket generalization \cite{cross} complies with ($k$, $l$)-anonymity principle to prevent identity disclosure and attribute disclosure. It not only provides greater security for sensitive attribute but also preserves more information utility than generalization when the demand of attribute protection is higher than that of identity protection.

\section{Conclusion and Future Study}
\label{sec_conclu}
This paper supposes that people can optionally set their sensitive values according to the personal requirements, and proposes a novel technique, namely, local generalization and bucketization, to provide secure protections for identities and sensitive values. The rationale is to divide the tuples into local equivalence groups and partition the sensitive values into local buckets through local generalization and local bucketization, respectively. Furthermore, the protections of local generalization and local bucketization are independent, thus their algorithms can be flexibly achieved according to practical requirements without weakening the other protection, respectively.

LGB is a flexible framework that can be used in many publishing scenarios for protecting confidential information. In the future study, we will combine with the approach of online learning that automatically protects the incremental publishing data. Additionally, the individual relationships, which is also considered as sensitive, are recommended to be studied in accordance with LGB.

\bibliographystyle{IEEEtran}
\bibliography{IEEEabrv,./references}

\begin{IEEEbiographynophoto}{Boyu Li}
	received his MS and PhD degrees in Computer Science and Technology from Jilin University, in 2014 and 2018, respectively. He is currently a post-doctor in College of Computer Science and Technology from Huazhong University of Science and Technology. His research interests include privacy-preserving data publishing.
\end{IEEEbiographynophoto}

\begin{IEEEbiographynophoto}{Kun He}
	is currently a Professor in School of Computer Science and Technology, Huazhong University of Science and Technology (HUST), Wuhan, P.R. China; and a Mary Shepard B. Upson Visiting Professor for the 2016-2017 Academic year in Engineering, Cornell University, NY, USA. She received the B.S degree in physics from Wuhan University, Wuhan, China, in 1993; the M.S. degree in computer science from Huazhong Normal University, Wuhan, China, in 2002; and the Ph.D. degree in system engineering from HUST, Wuhan, China, in 2006. Her research interests include: machine learning, deep learning, social networks, algorithm design and analysis.
\end{IEEEbiographynophoto}

\begin{IEEEbiographynophoto}{Geng Sun}
	received a BS degree in Communication Engineering from Dalian Polytechnic University, and a PhD degree in Computer Science and Technology from Jilin University, in 2007 and 2018, respectively. He was a visiting researcher in the School of Electrical and Computer Engineering at Georgia Institute of Technology, USA. He is currently a post-doctor at Jilin University. His research interests include wireless sensor networks, antenna array, collaborative beamforming and optimizations.
\end{IEEEbiographynophoto}

\end{document}